\documentclass[sigconf]{acmart}

\usepackage{booktabs} 

\usepackage{lipsum}
\usepackage{graphicx}
\usepackage[linesnumbered,ruled,vlined,hanginginout]{algorithm2e}
\usepackage{amssymb}
\usepackage{amsmath}
\usepackage{graphicx}
\usepackage{mathrsfs}
\usepackage{multicol}
\usepackage{longtable}
\usepackage{breqn}
\usepackage{hyperref}

\SetKwInput{KwInput}{Input}
\SetKwInput{KwOutput}{Output}

\newcommand{\algref}[1]{{\tt #1}}

\DeclareMathOperator{\Sol}{Sol}

\DeclareMathOperator{\Mon}{Mon}

\DeclareMathOperator{\lm}{lm}

\DeclareMathOperator{\NF}{NF}

\DeclareMathOperator{\ld}{ld}
\DeclareMathOperator{\init}{init}
\DeclareMathOperator{\res}{res}
\DeclareMathOperator{\disc}{disc}
\DeclareMathOperator{\sep}{sep}

\newcommand{\setderiv}{\Delta}

\newcommand{\automorphism}{\sigma}
\newcommand{\ranking}{\succ}
\newcommand{\setauto}{\Sigma}

\newcommand{\var}{y}
\newcommand{\algvar}{z}
\newcommand{\multclos}[1]{\langle \, #1 \, \rangle}

\newcommand{\diffidealgen}[1]{[#1]}

\newcommand{\Gr}{Gr\"obner }

\newcommand{\cI}{\mathcal{I}}
\newcommand{\cR}{\mathcal{R}}

\newcommand{\Q}{\mathbb{Q}}
\renewcommand{\C}{\mathbb{C}}
\newcommand{\Z}{\mathbb{Z}}

\newcommand{\R}{\mathbb{R}}

\newcommand{\cK}{\mathcal{K}}

\def \bg #1 {\begin{tabular}{{#1}}}
\def \nd {\end{tabular}}

\theoremstyle{acmdefinition}
\newtheorem{remark}[theorem]{Remark}

\acmDOI{}

\acmISBN{}
\setcopyright{acmcopyright}
\acmConference[ISSAC '19]{International Symposium on Symbolic and Algebraic Computation}{July 15-18, 2019}{Beijing, China}
\acmYear{2019}
\copyrightyear{2019}
\acmPrice{15.00}
\begin{document}
\title[Algorithmic approach to strong consistency analysis for FDA to PDE systems]{Algorithmic approach to strong consistency analysis of finite difference approximations to PDE systems}

\author{Vladimir~P.~Gerdt}
\affiliation{%
  \institution{ Joint Institute for Nuclear Research, Dubna, Russia \\ and Peoples' Friendship University of Russia (RUDN)}
  \city{Moscow}
  \state{Russia}
}
\email{gerdt@jinr.ru}

\author{Daniel Robertz}
\affiliation{%
  \institution{School of Computing, Electronics and Mathematics\\ University of Plymouth}
  \city{Plymouth}
  \state{United Kingdom}
}
\email{daniel.robertz@plymouth.ac.uk}

\begin{abstract}
For a wide class of polynomially nonlinear systems of partial differential equations we suggest an algorithmic approach to the s(trong)-consistency analysis of their finite difference approximations on Cartesian grids. First we apply the differential Thomas decomposition to the input system, resulting in a partition of the solution set. We consider the output simple subsystem that contains a solution of interest. Then, for this subsystem, we suggest an algorithm for verification of s-consistency for its finite difference approximation. For this purpose we develop a difference analogue of the differential Thomas decomposition, both of which jointly allow to verify the s-consistency of the 
approximation. As an application of our approach, we show how to produce s-consistent difference approximations to the incompressible Navier-Stokes equations including the pressure Poisson equation.

\end{abstract}



\maketitle

%
%
\section{Introduction}\label{sec:introduction}

Except very special cases, partial differential equations (PDE) admit numerical integration only.
Historically first and one of the most-used numerical methods is finite difference method~\cite{Samarskii'01} based on approximation of PDE by difference equations defined on a chosen solution grid. To construct a numerical solution, the obtained finite difference approximation (FDA) to PDE is augmented with an appropriate discretization of initial or/and boundary condition(s) providing uniqueness of solution. As this takes place, the quality of numerical\! solution to PDE\! is\! determined by the quality of its\! FDA.

Any reasonable discretization must provide the convergence of a numerical solution to a solution of PDE in the limit when the grid spacings tend to zero. However, except for a very limited class of problems, convergence cannot be directly established. In practice, for a given FDA, its consistency and stability are analyzed as the necessary conditions for convergence. Consistency implies reduction of the FDA to the original PDE when the grid spacings tend to zero and stability provides boundedness of the error in the solution under small perturbation in the numerical data.

One of the most challenging problems is to construct FDA which, on the one hand, approximates the PDE and, on the other hand, mimics basic algebraic properties and preserves the algebraic structure~\cite{Christiansen'11} of the PDE. Such mimetic or algebraic structure preserving FDA are more likely to produce highly accurate and stable numerical results (cf.~\cite{JCP'14}). In~\cite{GR'10,G'12}, for polynomially nonlinear PDE systems and regular solution grids, we introduced the novel concept of \emph{strong consistency}, or \emph{s-consistency}, which strengthens the concept of consistency and means that any element of the perfect difference ideal generated by the polynomials in FDA approximates an element in the
 radical differential ideal generated by the polynomials in PDE. In the subsequent papers~\cite{ABGLS'13,ABGLS'17}, by computational experiments with two-dimensional incompressible Navier-Stokes equations, it was shown that s-consistent FDA have much better numerical behavior than FDA which are not s-consistent.

For linear PDE one can algorithmically verify \cite{GR'10}
 s-consistency of their FDA. In the nonlinear case such verification~\cite{G'12} required computation of a difference \Gr basis for FDA. Since difference polynomial rings~\cite{Levin'08} are non-Noetherian, the difference \Gr basis algorithms~\cite{G'12,GLS'15} do not terminate in general.
 In comparison to differential algebra, fewer computational results have been obtained in difference algebra.
 A decomposition technique was developed only for binomial perfect difference ideals \cite{BinomialDifference}.
More generally, in the present paper, a difference analogue of the differential Thomas decomposition~\cite{Thomas'37-62,BGLHR'12,Robertz6,GLHR'18} is obtained (see Section~\ref{sec:Decomposition}), which provides an
algorithmic tool for s-consistency analysis of FDA to simple PDE subsystems
on Cartesian grids (see Section~\ref{sec:Discretization}). In particular, given an FDA to the momentum and continuity equations  in the Navier-Stokes PDE system for incompressible flow, our approach derives an s-consistent approximation containing the pressure Poisson equation (see Section~\ref{sec:NavierStokes}).

\noindent
 Completion to involution is the cornerstone of the differential Thomas decomposition~\cite{Thomas'37-62,BGLHR'12,Robertz6,GLHR'18}. The underlying completion algorithm~\cite{G'05} is based on the theory of Janet division and Janet bases \cite{G'05,Seiler'10,Robertz6} which stemmed from the Riquier-Janet theory~\cite{Riquier'10,Janet'29} of orthonomic PDE. Joseph M. Thomas~\cite{Thomas'37-62} generalized the Riquier-Janet theory to non-orthonomic polynomially nonlinear PDE and showed how to decompose them into the triangular subsystems with disjoint solution sets. Janet bases are \Gr ones with additional structure, and Wu Wen-tsun was the first who showed~\cite{Wu'90} that the Riquier-Janet theory can be used for algorithmic construction of algebraic \Gr bases. We dedicate this paper to commemoration of his Centennial Birthday.

\section{Consistency}\label{sec:consistency}

In the given paper we consider PDE systems of the form
\begin{equation}
f_1=\cdots=f_s=0\,,\quad F:=\{f_1,\ldots,f_s\}\subset {\cR}\,,\quad s\in \Z_{\geq 1}\,, \label{pde}
\end{equation}
where ${\cR}:={\cK}\{ \mathbf{u} \}$ is the ring of polynomials in the dependent variables $\mathbf{u}:=\{u^{(1)},\ldots,u^{(m)}\}$ and their partial derivatives
obtained from the operator power products in $\{\partial_1,\ldots,\partial_n\}$
$(\partial_j=\partial_{x_j})$. We shall assume that coefficients of the polynomials are rational functions in $\mathbf{a}:=\{a_1,\ldots,a_l\}$,
finitely many parameters (constants), over $\Q$,
i.e. ${\cK}:={\Q}(\mathbf{a})$. One can also extend the last field to ${\Q}(\mathbf{a},\mathbf{x})$, where $\mathbf{x}:=\{x_1,\ldots,x_n\}$ is the set of independent variables. In this case we shall assume that coefficients of the differential polynomials in $F$ do not vanish in the grid points defined in (\ref{grid}) below.

To approximate
(\ref{pde}) by a difference system we
define a Cartesian computational grid (mesh) with spacing $0 < h \in \R$
and fixed $\mathbf{x}$
by
\begin{equation}\label{grid}
\{\,(x_1+k_1h,\ldots,x_n+k_nh)\mid k_1, \ldots, k_n \in \Z\,\}\,,
\end{equation}
If the actual solution to
(\ref{pde}) is
$\mathbf{{u}}(\mathbf{x})$, then its approximation at the grid nodes will be denoted by
$\mathbf{\tilde{u}}_{k_1,\ldots,k_n}\approx \mathbf{u}(x_1+k_1h,\ldots,x_n+k_nh)$.

Let $\tilde{\cK} := \Q\,(\mathbf{a},h)$ and $\tilde{\cR}$ be
the \emph{difference polynomial ring} over $\tilde{\cK}$,
where $\tilde{\cK}$ is a difference field of constants~\cite{Levin'08} with differences $\setauto := \{\sigma_1,\ldots,\sigma_n\}$ acting on a grid function $\tilde{u}^{(\alpha)}_{k_1,\ldots,k_n}$ as the shift operators
\begin{equation}
\sigma_i^{\pm 1}
\tilde{u}^{(\alpha)}_{k_1,\ldots,k_i,\ldots,k_n} \, = \, \tilde{u}^{(\alpha)}_{k_1,\ldots,k_i\pm 1,\ldots,k_n}\,,\quad  \alpha\in \{1,\ldots,m\}\,.
\label{rs-operators}
\end{equation}
The elements in $\tilde{\cR}$ are polynomials in the dependent variables $\tilde{u}^{(\alpha)}$ ($\alpha=1,...,m$) defined on the grid and in their shifts $\sigma_1^{i_1}...\sigma_n^{i_n} \tilde{u}^{(\alpha)}$  $(i_j\in \Z)$.
However, to provide termination of the decomposition algorithm of Sect.~\ref{sec:Decomposition},
we shall consider difference polynomials with non-ne\-ga\-tive shifts only.
We denote by $\Mon(\setauto)$ the set of monomials in $\automorphism_1$, ..., $\automorphism_n$.
The coefficients of the polynomials are in $\tilde{\cK}$.

The standard method to obtain FDA of such type to the differential system \eqref{pde} is replacement of the partial derivatives occurring in \eqref{pde} by finite differences and application of appropriate power product of the forward-shift operators in \eqref{rs-operators} to eliminate negative
shifts in indices which may come out of
expressions like
\[
\partial_j u^{(\alpha)}(\mathbf{x}) \, = \, \frac{u^{(\alpha)}_{k_1,\ldots,k_j+1,\ldots,k_n}-u^{(\alpha)}_{k_1,\ldots,k_j-1,
\ldots,k_n}}{2h}+\mathcal{O}(h^2)\,.
\]
Furthermore, the difference system
\begin{equation}
\tilde{f}_1=\cdots=\tilde{f}_s=0\,,\quad \tilde{F}:=\{\tilde{f}_1,\ldots,\tilde{f}_s\}\subset
 \tilde{\cR}\,, \label{fda}
\end{equation}
is called an \emph{FDA} to PDE~\eqref{pde} if it is
consistent in accordance to:

\begin{definition}\label{def-wcons}
Given a PDE system~\eqref{pde}, a difference system~\eqref{fda} is \emph{weakly consistent} or \emph{w-consistent} with \eqref{pde} if
\begin{equation*}\label{w-consistency}
   (\,\forall\,j \in \{\,1,\ldots,s\,\}\,) \ \
 [\,\tilde{f}_j\xrightarrow[h\rightarrow 0]{}  f_j\,]\,.
\end{equation*}
\end{definition}
This is a universally adopted notion of consistency for a finite difference discretization of PDE system~\eqref{pde} (cf.~\cite{Str'04}, Ch.7) and means that Eq.~\eqref{fda} reduces to Eq.~\eqref{pde} when the mesh step $h$ tends to zero.

\begin{definition}{\cite{GR'10}}\,
We say that a difference equation $\tilde{f}(\mathbf{\tilde{u}})=0$, $\tilde{f}\in \tilde{\cR}$, implies the differential equation
  $f(\mathbf{\mathbf{u}})=0$, $f\in {\cR}$,
and write
$\tilde{f}\rhd f$,
if the Taylor expansion of $\tilde{f}$ about the grid point $\mathbf{x}$, after clearing denominators containing $h$, yields
\begin{equation}\label{implication}
  \tilde{f}(\mathbf{\tilde{u}})=h^d f(\mathbf{u})+\mathcal{O}\,(\,h^{d+1}\,)\,,\quad d\in \Z_{\geq 0}\,,
\end{equation}
and $\mathcal{O}\,(\,h^{d+1}\,)$ denotes terms whose degree in $h$ is at least $d+1$.
\label{DefImplication}
\end{definition}

\begin{remark}\label{rem:continuouslimit}
Given $\tilde{f}(\tilde{\mathbf{u}})$, computation of $f(\mathbf{u})$ is straightforward and has been implemented as routine {\sc ContinuousLimit} in the Maple package {\sc LDA} \cite{GR'12,GLS'15}  ({\underline{L}inear \underline{D}ifference \underline{A}lgebra}).
\end{remark}

\begin{definition}{\cite{G'12}}
FDA~\eqref{fda} to PDE system~\eqref{pde} is \emph{strongly consistent} or \emph{s-consistent} if
\begin{equation}\label{def-scon}
(\,\forall \tilde{f}\in \llbracket \tilde{F} \rrbracket\,)\  (\,\exists\, f\in \llbracket F \rrbracket\,)\  \ [\,\tilde{f}\rhd f\,]\,.
\end{equation}
Here $\llbracket \tilde{F} \rrbracket$ and $\llbracket F \rrbracket$ denote
the perfect difference ideal generated by $\tilde{F}$ in $\tilde{\cR}$ and the radical differential ideal generated by $F$ in ${\cR}$.
\label{DefSC}
\end{definition}

\begin{remark}
It is clear that if condition~\eqref{implication} holds, then
 \begin{equation}\label{approximation}
  \frac{\Tilde{f}(\mathbf{\tilde{u}})}{h^d}\xrightarrow[h\rightarrow 0]{}f(\mathbf{u})\,,
 \end{equation}
 that is, $\tilde{f}(\mathbf{\tilde{u}})/h^d$ approximates $f(\mathbf{u})$. Accordingly, condition~\eqref{def-scon} means that, after clearing denominators, each element of $\llbracket \tilde{F} \rrbracket$ approximates an element of $\llbracket F \rrbracket$ in the sense of (\ref{approximation}).
\end{remark}

\begin{lemma}\label{lem:perfectclosure}
Let ${\cI}=[{F}]$ be a differential ideal of ${\cR}$ and $\tilde{\cI}=[\tilde{F} ]$ a difference ideal of
$\tilde{\cR}$ such that
\[
(\,\forall \tilde{f}\in \tilde{\cI} \,)\  (\,\exists\, f\in {\cI} \,)\  \ [\,\tilde{f}\rhd f\,]\,.
\]
Then for the perfect closure $\llbracket \tilde{\cI} \rrbracket$ of $\tilde{\cI}$ in $\tilde{\cR}$ the condition~\eqref{def-scon} holds.
\end{lemma}

\begin{proof}
Let $\tilde{G}$ be a (possibly infinite) reduced \Gr basis of $\tilde{\cI}$ for an admissible monomial ordering $\succ$ (cf.~\cite{G'12}). Then
\[
\tilde{f}=\sum_{\tilde{g}\in \tilde{G}_1\subseteq \tilde{G}} \sum_{\mu} a_{\tilde{g},\mu} \sigma^\mu {\tilde{g}}\,,\quad a_{\tilde{g},\mu}\in \tilde{\cR}\,,\quad
\lm(a_{\tilde{g},\mu} \sigma^\mu {\tilde{g}})\preceq \lm(\tilde{f})\,. \label{sb-ideal}
\]
Here $\tilde{f}\in \tilde{\cI}$, $\tilde{G}_1$ is a finite subset of $\tilde{G}$, $\lm$ denotes the {\em leading monomial} of its argument, and we use the multi-index notation
\[
\mu:=(\mu_1,\ldots,\mu_n)\in \Z^n_{\geq 0}\,, \quad \sigma^\mu:=\sigma_1^{\mu_1}\cdots \sigma_n^{\mu_n}\,.
\]
In the continuous limit $\tilde{f}$ implies the differential polynomial
\[
 f:=\sum_{g\in G_1} \sum_{\nu}b_{g,\nu} \partial^\nu g\,,\quad b_{g,\nu}\in {\cR}\,,
\]
where $\tilde{G}_1\rhd G_1$. Therefore, $\tilde{f} \rhd f\in [F]\subseteq \llbracket F \rrbracket$.

Let now $\tilde{p}\in \llbracket \tilde{F} \rrbracket \setminus [\tilde{F}]$ and $\theta_1,\ldots,\theta_r\in \Mon(\setauto)$  be such that
\begin{equation}
\tilde{q}:=(\theta_1 \tilde{p})^{k_1}\cdots (\theta_r \tilde{p})^{k_r}\in [\tilde{F}]\,, \quad k_1,\ldots,k_r \in \Z_{\geq 0}\,.
\label{shuffling}
\end{equation}
From~Eq.~\eqref{shuffling} it follows that $\tilde{q} \, \rhd \, q=p^{k_1+\cdots+ k_r}$ where $\tilde{p} \, \rhd \, p$. Hence, $p\in \llbracket F\rrbracket$. The perfect ideal $\llbracket \tilde{F} \rrbracket$ can be constructed \cite{Levin'08} from $[\tilde{F}]$ by the procedure in the form called {\em shuffling} and based on enlargement of the generator set $\tilde{F}$ with all polynomials $\tilde{p}$ occurring in $[\tilde{F}]$ in the form of Eq.~\eqref{shuffling} and on repetition of such enlargement. It is clear that each such enlargement of the intermediate ideals yields in the continuous limit a subset of $\llbracket F \rrbracket$.
\end{proof}

\noindent
The criterion of s-consistency is given by the following theorem.

\begin{theorem}{\em\cite{G'12}}
A difference approximation (\ref{fda}) to a differential system (\ref{pde}) is s-consistent if and only if a reduced \Gr basis  $\tilde{G}\subset \tilde{\cR}$ of the difference
ideal $[\tilde{F}]\subset \tilde{\cR}$ generated by $\tilde{F}$ satisfies
\begin{equation*}
(\,\forall \tilde{g}\in \tilde{G}\,)\ (\,\exists\, g\in \llbracket F \rrbracket\,)\  [\,\tilde{g}\rhd g\,]\,. \label{cons-gb}
\end{equation*}
\label{th:s-cons}
\end{theorem}
\vspace{-1em}

%
%
\section{Janet division}\label{sec:janetdivision}

We recall the concept of Janet division.
For details we refer to, e.g., \cite[Subsect.~2.1.1]{Robertz6}, \cite{G'05}, \cite[Ch.~3]{Seiler'10}.

Let $K$ be a field and $R := K[\var_1, \ldots,
 \var_n]$
the commutative polynomial algebra over $K$ with indeterminates
$\var_1$, \ldots, $\var_n$. We denote by $\Mon(R)$ the set of monomials
in $\var_1$, \ldots, $\var_n$ and for a subset $\mu \subseteq \{ \var_1, ..., \var_n \}$
we define $\Mon(\mu)$ to be the subset of $\Mon(R)$ consisting of the monomials involving only
indeterminates from $\mu$.

If a term ordering on $R$ is fixed and $I$ is an ideal of $R$,
then the set of leading monomials of non-zero polynomials in $I$ are known to form a set with the
following property:

\begin{definition}
A set $M \subseteq \Mon(R)$ is said to be
\emph{$\Mon(R)$-multiple-closed} if we have $r m \in M$
for all $m \in M$ and all $r \in \Mon(R)$.
\end{definition}

The smallest $\Mon(R)$-multiple-closed set in $\Mon(R)$ containing
a given set $G \subseteq \Mon(R)$ is denoted by $\multclos{G}$.
It is well known that every $\Mon(R)$-multiple-closed set in $\Mon(R)$
is finitely generated in that sense and that it has a unique
minimal generating set.

We adopt Janet's approach \cite{Janet'29}
of partitioning
a $\Mon(R)$-multiple-closed set $M$ into finitely many
subsets of the form $\Mon(\mu) \, m$, where $m \in M$ and $\mu = \mu(m, M) \subseteq \Mon(R)$
(referred to as Janet division).

\begin{definition}
Let $G \subset \Mon(R)$ be finite and $m = \var_1^{i_1} \cdots \var_n^{i_n} \in G$. Then
$\var_k$ is said to be a \emph{multiplicative variable} for $m$ if and only if
\[
i_k = \max \, \{ \, j_k \mid \var_1^{j_1} \cdots \var_n^{j_n} \in G \mbox{ with }
j_1 = i_1, \, \ldots, \, j_{k-1} = i_{k-1} \, \}\,.
\]
This yields
a partition $\{ \var_1, \ldots, \var_n \} = \mu(m, G) \uplus \overline{\mu}(m, G)$,
where the elements of $\mu(m, G)$ (resp.\ $\overline{\mu}(m, G)$) are the
multiplicative (resp.\ non-mul\-ti\-pli\-ca\-tive) variables for $m$.
The set $G$ is \emph{Janet complete} if
\[
\multclos{G} := \bigcup_{m \in G} \Mon(R) \, m = \biguplus_{m \in G} \Mon(\mu(m, G)) \, m\,.
\]
\end{definition}

\begin{proposition}\label{prop:Janetcompletion}
For every $\Mon(R)$-multiple-closed set $M$
there exists a finite Janet complete set $J \subset \Mon(R)$
such that $M = \multclos{J}$.
\end{proposition}

If $G \subset \Mon(R)$ is finite, we call
the minimal Janet complete set $J \supset G$ such that $\multclos{J} = \multclos{G}$
the \emph{Janet completion} of $G$. It is obtained algorithmically
by adding certain multiples of elements of $G$ to $G$
(which also proves Proposition~\ref{prop:Janetcompletion}),
cf., e.g., \cite[Algorithm~2.1.6]{Robertz6}.

%
%
\section{Simple Algebraic Systems}\label{sec:simplealgebraicsystems}

Fundamental for both the differential Thomas decomposition (recalled in Sect.~\ref{sec:differentialthomas})
as well as its difference analogue to be introduced in Sect.~\ref{sec:Decomposition}
is the Thomas decomposition of an \emph{algebraic system} $S$
\begin{equation}\label{eq:algebraicsystem}
p_1 = 0\,, \ \ldots\,, \ p_s = 0\,, \ p_{s+1} \neq 0\,, \ \ldots\,, \ p_{s+t} \neq 0 \quad (s, t \in \Z_{\ge 0})
\end{equation}
where $p_1$, \ldots, $p_{s+t} \in R := K[\algvar_1, \ldots, \algvar_n]$.
Here $K$ is a field of characteristic zero with algebraic closure $\overline{K}$, and $R$ is
the commutative polynomial algebra over $K$
with indeterminates $\algvar_1$, \ldots, $\algvar_n$. The solution set of
the algebraic system $S$ in (\ref{eq:algebraicsystem})
is defined to be
\[
\Sol_{\overline{K}}(S) := \{ a \in \overline{K}^n \mid p_i(a) = 0, \, p_{s+j}(a) \neq 0, \, i = 1, ..., s, \, j = 1, ..., t \}\,.
\]
Assuming the indeterminates are ordered as in $\algvar_1 \ranking \algvar_2 \ranking \ldots \ranking \algvar_n$,
a sequence of projections from $\overline{K}^n$ is defined correspondingly by
\[
\pi_i\colon \overline{K}^n \to \overline{K}^{n-i}\colon (a_1, ..., a_n) \mapsto (a_{i+1}, ..., a_n)\,, \, \,
i = 1, 2, ..., n-1\,.
\]
For each $p \in R \setminus K$, this ordering defines
the greatest indeterminate $\ld(p)$ occurring in $p$, referred to as \emph{leader},
the coefficient $\init(p)$ of the highest power of $\ld(p)$ in $p$, called \emph{initial},
and the \emph{discriminant} $\disc(p) := (-1)^{d(d-1)/2} \res(p, \partial p / \partial \ld(p), \ld(p)) \, / \, \init(p)$,
where $d$ is the degree of $p$ in $\ld(p)$ and where $\res$ denotes the resultant.

\begin{definition}\label{de:algebraicsimple}
An algebraic system $S$ as in (\ref{eq:algebraicsystem})
is said to be \emph{simple} if the following four conditions are satisfied.
\begin{enumerate}
\item None of $p_1$, \ldots, $p_s$, $p_{s+1}$, \ldots, $p_{s+t}$ is constant.\label{de:algebraicsimple_1}
\item The leaders of $p_1$, ..., $p_s$, $p_{s+1}$, ..., $p_{s+t}$ are pairwise distinct.\label{de:algebraicsimple_2}
\item For every $r \in \{ p_1, \ldots, p_s, p_{s+1}, \ldots, p_{s+t} \}$,
if $\ld(r) = \algvar_k$,
then the equation $\init(r) = 0$
has no solution in $\pi_k(\Sol_{\overline{K}}(S))$.\label{de:algebraicsimple_3}
\item For every $r \in \{ p_1, \ldots, p_s, p_{s+1}, \ldots, p_{s+t} \}$,
if $\ld(r) = \algvar_k$,
then the equation $\disc(r) = 0$
has no solution in $\pi_k(\Sol_{\overline{K}}(S))$.\label{de:algebraicsimple_4}
\end{enumerate}
(In (\ref{de:algebraicsimple_3}) and (\ref{de:algebraicsimple_4}), we have $\init(r)$, $\disc(r) \in K[\algvar_{k+1}, \ldots, \algvar_n]$.)
\end{definition}

\begin{definition}\label{de:algebraicquasisimple}
An algebraic system $S$ as in (\ref{eq:algebraicsystem}) is
said to be \emph{quasi-simple} if conditions~(\ref{de:algebraicsimple_1})--(\ref{de:algebraicsimple_3})
(but not necessarily (\ref{de:algebraicsimple_4})) are satisfied.
\end{definition}

A \emph{Thomas decomposition} of an algebraic system $S$ as in (\ref{eq:algebraicsystem})
is a finite collection of simple algebraic systems $S_1$, ..., $S_r$
such that $\Sol_{\overline{K}}(S)\!=\!\Sol_{\overline{K}}(S_1) \uplus ... \uplus \Sol_{\overline{K}}(S_r)$.
It can be computed by an algorithm combining Euclidean pseudo-reduction and case distinctions.
For details we refer to \cite{BGLHR'12}, \cite[Subsect.~2.2.1]{Robertz6}, \cite[Sect.~3.3]{Wang'00}.

%
%
\section{Differential Thomas Decomposition}\label{sec:differentialthomas}

A \emph{system\! of polynomial partial differential equations and inequations}
\begin{equation}\label{eq:polypartialdifferential}
f_1 = 0\,, \ \ldots\,, \ f_s = 0\,, \ f_{s+1} \neq 0\,, \ \ldots\,, \ f_{s+t} \neq 0 \quad (s, t \in \Z_{\ge 0})
\end{equation}
is given by
elements $f_1$, \ldots, $f_{s+t}$ of the differential polynomial ring ${\cR}$
in $u^{(1)}$, \ldots, $u^{(m)}$ with commuting derivations
$\setderiv := \{ \partial_1, \ldots, \partial_n \}$.
For $\alpha \in \{ 1, \ldots, m \}$, $J \in (\Z_{\ge 0})^n$ we identify $u^{(\alpha)}_J$
and $\partial_1^{J_1} \cdots \partial_n^{J_n} u^{(\alpha)}$.
Let $\Omega \subseteq \R^n$ be open and connected. The solution set of $S$ on $\Omega$ is
\[
\begin{array}{c}
\Sol_{\Omega}(S) := \{ \, a = (a_1, \ldots, a_m) \mid a_k\colon \Omega \to \C \mbox{ analytic}, \, k = 1, \ldots, m,\\[0.2em]
f_i(a) = 0, \, f_{s+j}(a) \neq 0, \, i = 1, \ldots, s, \, j = 1, \ldots, t \, \}\,.
\end{array}
\]

\begin{definition}\label{de:ranking}
A \emph{ranking} $\ranking$ on ${\cR}$ is a total ordering on the set
\[
\Mon(\setderiv) \, u \, := \, \{ \, u^{(\alpha)}_J \mid 1 \le \alpha \le m, \, J \in (\Z_{\ge 0})^n \, \}
\]
such that for all $j \in \{ 1, \ldots, n \}$, $\alpha$, $\beta \in \{ 1, \ldots, m \}$,
$J$, $K \in (\Z_{\ge 0})^n$
we have $\partial_j u^{(\alpha)} \ranking u^{(\alpha)}$ and,
if $u^{(\alpha)}_{J} \ranking u^{(\beta)}_{K}$, then
$\partial_j u^{(\alpha)}_{J} \ranking \partial_j u^{(\beta)}_{K}$.
A ranking $\ranking$ is \emph{orderly} if
for all $\alpha$, $\beta \in \{ 1, \ldots, m \}$, $J$, $K \in (\Z_{\ge 0})^n$,
$J_1 + \cdots + J_n > K_1 + \cdots + K_n$
implies $u^{(\alpha)}_{J} \ranking u^{(\beta)}_{K}$.
\end{definition}

\begin{example}\label{ex:lexrankings}
Rankings $\ranking_{{\rm TOP},{\rm lex}}$
and $\ranking_{{\rm POT},{\rm lex}}$
on ${\cR}$ are given by
\[
u^{(\alpha)}_J \ranking_{{\rm TOP},{\rm lex}} u^{(\beta)}_K \quad
:\Leftrightarrow \quad
J \ranking_{{\rm lex}} K \, \, \mbox{or} \, \, (\, J = K \mbox{ and } \alpha < \beta \, )
\]
and
\[
u^{(\alpha)}_J \ranking_{{\rm POT},{\rm lex}} u^{(\beta)}_K \quad
:\Leftrightarrow \quad
\alpha < \beta \, \, \mbox{or} \, \, (\, \alpha = \beta \mbox{ and }
J \ranking_{{\rm lex}} K \, )\,,
\]
respectively, where $\ranking_{{\rm lex}}$ compares multi-indices lexicographically.
\end{example}

If a ranking $\ranking$ on ${\cR}$ is fixed, then for each
$f \in {\cR} \setminus {\cK}$ the \emph{leader}, \emph{initial} and
\emph{discriminant} of $f$ are defined as in Section~\ref{sec:simplealgebraicsystems}.
Moreover,
$\sep(f) := \partial f / \partial \ld(f)$
is called the \emph{separant} of $f$.

Janet division associates (with respect to a total ordering of $\setderiv$)
to each $f_i = 0$ with $\ld(f_i) = \theta_i u^{(\alpha)}$
the set $\mu_i := \mu(\theta_i, G_{\alpha}) \subseteq \setderiv$ (resp.\ $\overline{\mu}_i := \setderiv \setminus \mu_i$) of
\emph{admissible} (resp.\
\emph{non-admissible}) \emph{derivations}, where
\[
G_{\alpha} := \{ \, \theta \in \Mon(\setderiv) \mid \theta u^{\alpha} \in \{ \ld(f_1), \ldots, \ld(f_s) \} \, \}\,.
\]
We call $\{ f_1 \!=\! 0, ..., f_s \!=\! 0 \}$ or $T := \{ (f_1, \mu_1), ..., (f_s, \mu_s) \}$
\emph{Janet complete}
if each $G_{\alpha}$ equals its Janet completion, $\alpha = 1$, \ldots, $m$.
Let $r \in {\cR}$.
If some $v \in \Mon(\setderiv) u$ occurs in $r$ for which there exists
$(f, \mu) \in T$ such that $v = \theta \ld(f)$ for some $\theta \in \Mon(\mu)$
and $\deg_v(r) \ge \deg_v(\theta f)$, then
$r$ is \emph{Janet reducible modulo $T$}.
In this case, $(f, \mu)$ is called a
\emph{Janet divisor} of $r$. If $r$ is not Janet reducible modulo $T$,
then $r$ is also said to be \emph{Janet reduced modulo $T$}.
Iterated pseudo-re\-duc\-tions of $r$ modulo $T$
yield its \emph{Janet normal form} $\NF(r, T, \ranking)$, 
a Janet reduced
differential polynomial, as explained in \cite[Algorithm~2.2.45]{Robertz6}.

\begin{definition}\label{de:differentialpassive}
Let $T = \{ \, (f_1, \mu_1), \ldots, (f_s, \mu_s) \, \}$
be Janet complete.
Then $\{ \, f_1 = 0, \ldots, f_s = 0 \, \}$ or $T$ is said to be \emph{passive}, if
\[
\NF(\partial f_i, T, \ranking) \, = \, 0 \qquad \mbox{for all} \quad
\partial \in \overline{\mu}_i = \setderiv \setminus \mu_i\,, \quad i = 1, \ldots, s\,.
\]
\end{definition}

\begin{definition}\label{de:differentialsimple}
Let a ranking $\ranking$ on ${\cR}$ and a total ordering on $\setderiv$ be fixed.
A differential system $S$ as in (\ref{eq:polypartialdifferential})
is said to be \emph{simple} if the following three conditions hold.
\begin{enumerate}
\item $S$ is simple as an algebraic system
(in the finitely many indeterminates occurring in it, ordered by
the ranking $\ranking$).\label{de:differentialsimple_a}
\item $\{ \, f_1 = 0, \, \ldots, \, f_s = 0 \, \}$ is passive.\label{de:differentialsimple_b}
\item The left hand sides $f_{s+1}$, \ldots, $f_{s+t}$
are Janet reduced modulo
the passive differential system $\{ \, f_1 = 0, \, \ldots, \, f_s = 0 \, \}$.\label{de:differentialsimple_c}
\end{enumerate}
\end{definition}

\begin{proposition}[\cite{Robertz6}, Prop.~2.2.50]\label{prop:differentialmembership}
Let $S$ be a simple differential system, defined over ${\cR}$,
as in (\ref{eq:polypartialdifferential}).
Let $E$ be the
differential ideal of ${\cR}$ which is generated by $f_1$, \ldots, $f_s$ and
let $q$ be the product of the initials and separants of all $f_1$,
\ldots, $f_s$. Then the differential ideal
\[
E : q^{\infty} := \{ \, f \in {\cR} \mid q^r \, f \in E \mbox{ for some } r \in \Z_{\ge 0} \, \}
\]
is radical.
Given $f \in {\cR}$, we have $f \in E : q^{\infty}$ if and only if the
Janet normal form
of $f$ modulo $f_1$, \ldots, $f_s$
is zero.
\end{proposition}

\begin{definition}\label{DifferentialThomasDecomposition}
A \emph{Thomas decomposition} of a differential system $S$ as in (\ref{eq:polypartialdifferential})
(with respect to $\ranking$)
is a finite collection of simple differential systems $S_1$, ..., $S_r$ such that
$\Sol_{\Omega}(S) = \Sol_{\Omega}(S_1) \uplus ... \uplus \Sol_{\Omega}(S_r)$.
\end{definition}

For any differential system $S$ as in (\ref{eq:polypartialdifferential}) and any
ranking $\ranking$ on ${\cR}$ a Thomas decomposition of $S$ can be computed algorithmically.
For more details we refer to, e.g., \cite{BGLHR'12}, \cite[Subsection~2.2.2]{Robertz6},
\cite{GLHR'18}.

%
%
\section{\mbox{Decomposition of difference systems}}\label{sec:Decomposition}

A \emph{system $\tilde{S}$ of polynomial partial difference equations and inequations}
\begin{equation}\label{eq:polypartialdifference}
\tilde{f}_1 = 0\,, \ \ldots\,, \ \tilde{f}_s = 0\,, \ \tilde{f}_{s+1} \neq 0\,, \ \ldots\,, \ \tilde{f}_{s+t} \neq 0 \quad (s, t \in \Z_{\ge 0})
\end{equation}
is given by
elements $\tilde{f}_1$, \ldots, $\tilde{f}_{s+t}$
of the difference polynomial ring $\tilde{\cR}$
in $\tilde{u}^{(1)}$, \ldots, $\tilde{u}^{(m)}$
with commuting automorphisms
$\setauto = \{ \automorphism_1, \ldots, \automorphism_n \}$.
For $\alpha \in \{ 1, \ldots, m \}$, $J \in (\Z_{\ge 0})^n$ we identify $\tilde{u}^{(\alpha)}_J$
and $\automorphism_1^{J_1} \cdots \automorphism_n^{J_n} \tilde{u}^{(\alpha)}$.
We denote by $\tilde{S}^{=}$ (resp.\ $\tilde{S}^{\neq}$) the set
$\{ \tilde{f}_1, ..., \tilde{f}_s \}$ (resp.\ $\{ \tilde{f}_{s+1}, ..., \tilde{f}_{s+t} \}$).

A \emph{ranking} on $\tilde{\cR}$ is defined in the same way
as in Definition~\ref{de:ranking} by replacing the action
of $\partial_i$ by the action of $\automorphism_i$
and $\setderiv$ by $\setauto$.

For a subset $L$ of $\tilde{\cR}$ we denote by $\diffidealgen{L}$ the difference ideal of $\tilde{\cR}$ generated by $L$.
Let $E$ be a difference ideal of $\tilde{\cR}$
and $\emptyset \neq Q \subseteq \tilde{\cR}$ be multiplicatively closed and closed under $\automorphism_1$, \ldots, $\automorphism_n$. Then define
\[
E : Q \, := \, \{ \, \tilde{f} \in \tilde{\cR} \mid q \, \tilde{f} \in E \mbox{ for some } q \in Q \, \}\,.
\]
Moreover, for $U \subseteq \Mon(\setauto) \, \tilde{u}$ and $v \in \Mon(\setauto) \, \tilde{u}$ we define
\[
U : v \, := \, \{ \, \theta \in \Mon(\setauto) \mid \theta \, v \in U \, \}\,.
\]

The first algorithm to be introduced performs an auto-reduction of a finite set of difference polynomials.

\begin{algorithm}
\DontPrintSemicolon
\KwInput{$L \subset \tilde{\cR} \setminus \tilde{\cK}$ finite and a ranking $\ranking$ on $\tilde{\cR}$ such that $L = \tilde{S}^{=}$ for some finite difference system $\tilde{S}$ which is quasi-simple as an algebraic system (in the finitely many indeterminates $\tilde{u}^{(\alpha)}_J$ which occur in it, totally ordered by $\ranking$)
}
\KwOutput{$a \in \{ \text{\bf true}, \text{\bf false} \}$ and $L' \subset \tilde{\cR} \setminus \tilde{\cK}$ finite such that
\begin{equation*}\label{eq:LstrichL}
\diffidealgen{L'} : Q = \diffidealgen{L} : Q\,,
\end{equation*}
where $Q$ is the smallest multiplicatively closed subset of $\tilde{\cR}$ containing all $\init(\theta \tilde{f})$, where $\tilde{f} \in L$
and $\theta \in \ld(L \setminus \{ \tilde{f} \}) : \ld(\tilde{f})$, and which is closed under $\automorphism_1$, \ldots, $\automorphism_n$,
and, in case $a = \text{\bf true}$, there exist no $\tilde{f}_1$, $\tilde{f}_2 \in L'$, $\tilde{f}_1 \neq \tilde{f}_2$, such that we have $v := \ld(\tilde{f}_1) = \theta \ld(\tilde{f}_2)$ for some $\theta \in \Mon(\setauto)$ and $\deg_{v}(\tilde{f}_1) \ge \deg_{v}(\theta \tilde{f}_2)$}
$L' \gets L$\;
\While{$\exists \, \tilde{f}_1, \tilde{f}_2 \in L', \, \tilde{f}_1 \neq \tilde{f}_2$ and $\theta \in \Mon(\setauto)$ such that we have $v := \ld(\tilde{f}_1) = \theta \ld(\tilde{f}_2)$ and $\deg_{v}(\tilde{f}_1) \ge \deg_{v}(\theta \tilde{f}_2)$}
{
  $L' \gets L' \setminus \{ \tilde{f}_1 \}$; \, $v \gets \ld(\tilde{f}_1)$\;
  $\tilde{r} \! \gets \! \init(\theta \tilde{f}_2) \, \tilde{f}_1 - \init(\tilde{f}_1) \, v^d \, \theta \tilde{f}_2$, $d \! := \! \deg_{v}(\tilde{f}_1) - \deg_{v}(\theta \tilde{f}_2)$\;
  \If{$\tilde{r} \neq 0$}{
    \Return $(\text{\bf false}, L' \cup \{ \tilde{r} \})$\;
  }
}
\Return $(\text{\bf true}, L')$\;
\caption{\algref{Auto-reduce} for difference algebra\label{alg:autoreducenonlin}}
\end{algorithm}

Since leaders are dealt with in decreasing order with respect to $\ranking$, and no ranking admits infinitely decreasing chains, Algorithm~\ref{alg:autoreducenonlin} \emph{terminates}.
Its \emph{correctness} follows from the definition of $E : Q$.
\smallskip

Janet division associates (with respect to a total ordering of $\setauto$)
to each $\tilde{f}_i = 0$ with $\ld(\tilde{f}_i) = \theta_i \tilde{u}^{(\alpha)}$
the set $\mu_i := \mu(\theta_i, \tilde{G}_{\alpha}) \subseteq \setauto$ (resp.\ $\overline{\mu}_i := \setauto \setminus \mu_i$) of
\emph{admissible}
(\emph{non-admissible}) \emph{automorphisms}, where
\[
\tilde{G}_{\alpha} := \{ \, \theta \in \Mon(\setauto) \mid \theta \tilde{u}^{\alpha} \in \{ \ld(\tilde{f}_1), \ldots, \ld(\tilde{f}_s) \} \, \}\,.
\]
We call $\{ \tilde{f}_1 \!=\! 0, ..., \tilde{f}_s \!=\! 0 \}$ or $T := \{ (\tilde{f}_1, \mu_1), ..., (\tilde{f}_s, \mu_s) \}$
\emph{Janet complete}
if each $\tilde{G}_{\alpha}$ equals its Janet completion, $\alpha = 1$, \ldots, $m$.
Let $\tilde{r} \in \tilde{\cR}$.
If some $v \in \Mon(\setauto) \tilde{u}$ occurs in $\tilde{r}$ for which there exists
$(\tilde{f}, \mu) \in T$ such that $v = \theta \ld(\tilde{f})$ for some $\theta \in \Mon(\mu)$
and $\deg_v(\tilde{r}) \ge \deg_v(\theta \tilde{f})$, then
$\tilde{r}$ is \emph{Janet reducible modulo $T$}.
In this case, $(\tilde{f}, \mu)$ is called a
\emph{Janet divisor} of $\tilde{r}$. If $\tilde{r}$ is not Janet reducible modulo $T$,
then $\tilde{r}$ is also said to be \emph{Janet reduced modulo $T$}.
Iterated pseudo-re\-duc\-tions of $\tilde{r}$ modulo $T$
yield its \emph{Janet normal form} $\NF(\tilde{r}, T, \ranking)$, which is the Janet reduced
difference polynomial $\tilde{r}'$ returned by Algorithm~\ref{alg:janetreducenonlin}.

\begin{algorithm}
\DontPrintSemicolon
\KwInput{$\tilde{r} \in \tilde{\cR}$, $T = \{ \, (\tilde{f}_1, \mu_1), (\tilde{f}_2, \mu_2), \ldots, (\tilde{f}_s, \mu_s) \, \}$, and a ranking $\ranking$ on $\tilde{\cR}$, where $T$ is Janet complete (with respect to $\ranking$)}
\KwOutput{$(\tilde{r}', b) \in \tilde{\cR} \times \tilde{\cR}$ such that
(1) if $\tilde{r} \in \tilde{\cK}$ or $T = \emptyset$, then $\tilde{r}' = \tilde{r}$, $b = 1$,
(2) otherwise $\tilde{r}'$ is Janet-reduced modulo $T$ and
$$
\tilde{r}' + \diffidealgen{\tilde{f}_1, \ldots, \tilde{f}_s} = b \cdot \tilde{r} + \diffidealgen{\tilde{f}_1, \ldots, \tilde{f}_s}\,,
$$
where $b$ is
in the multiplicatively closed set generated by
$$
\bigcup_{i=1}^s \, \{ \, \theta \init(\tilde{f}_i) \mid \theta \in \Mon(\setauto), \, \ld(\tilde{r}) \ranking \theta \ld(\tilde{f}_i) \, \} \cup \{ 1 \}
$$}
$\tilde{r}' \gets \tilde{r}$; \, $b \gets 1$\;
\If{$\tilde{r}' \not\in \tilde{\cK}$}
{
  $v \gets \ld(\tilde{r}')$\;
  \While{$\tilde{r}' \not\in \tilde{\cK}$ and there exist $(\tilde{f}, \mu) \in T$ and $\theta \in \Mon(\mu)$ such that $v = \theta \ld(\tilde{f})$ and $\deg_{v}(\tilde{r}') \ge \deg_{v}(\theta \tilde{f})$}
  {
    $\tilde{r}' \! \gets \! \init(\theta \tilde{f}) \, \tilde{r}' \! - \! \init(\tilde{r}') \, v^d \, \theta \tilde{f}$, $d \! := \! \deg_{v}(\tilde{r}') \! - \! \deg_{v}(\theta \tilde{f})$\;
    $b \gets \init(\theta \tilde{f}) \cdot b$\;
  }
  \For{each coefficient $\tilde{c}$ of $\tilde{r}'$ (as a polynomial in $v$)}
  {
    $(\tilde{r}'', b') \gets$ \algref{Janet-reduce}($\tilde{c}$, $T$, $\ranking$)\;
    replace the coefficient $b' \cdot \tilde{c}$ in $b' \cdot \tilde{r}'$ with $\tilde{r}''$ and replace $\tilde{r}'$ with this result\;
    $b \gets b' \cdot b$\;
  }
}
\Return $(\tilde{r}', b)$\;
\caption{\algref{Janet-reduce} for difference algebra\label{alg:janetreducenonlin}}
\end{algorithm}

Algorithm~\ref{alg:janetreducenonlin} \emph{terminates} because each coefficient $\tilde{c}$ of $\tilde{r}'$ is either constant or has a leader which is smaller than $\ld(\tilde{r}')$ with respect to $\ranking$, and a ranking $\ranking$ does not allow infinitely decreasing chains.
\emph{Correctness} of the algorithm is clear.

\begin{definition}\label{de:differencepassive}
Let $T = \{ \, (\tilde{f}_1, \mu_1), \ldots, (\tilde{f}_s, \mu_s) \, \}$
be Janet complete.
Then $\{ \, \tilde{f}_1 = 0, \ldots, \tilde{f}_s = 0 \, \}$ or $T$ is said to be \emph{passive}, if
\[
\NF(\automorphism \tilde{f}_i, T, \ranking) = 0 \qquad \mbox{for all} \quad
\automorphism \in \overline{\mu}_i = \setauto \setminus \mu_i\,, \quad i = 1, \ldots, s\,.
\]
\end{definition}

\begin{definition}\label{de:differencesimple}
Let a ranking $\ranking$ on $\tilde{\cR}$ and a total ordering on $\setauto$ be fixed.
A difference system $\tilde{S}$ as in (\ref{eq:polypartialdifference})
is said to be \emph{simple} (resp., \emph{quasi-simple}) if the following three conditions hold.
\begin{enumerate}
\item $\tilde{S}$ is simple (resp., quasi-simple) as an algebraic system
(in the finitely many occurring indeterminates, ordered by $\ranking$).\label{de:differencesimple_a}
\item $\{ \, \tilde{f}_1 = 0, \, \ldots, \, \tilde{f}_s = 0 \, \}$ is passive.\label{de:differencesimple_b}
\item The left hand sides $\tilde{f}_{s+1}$, \ldots, $\tilde{f}_{s+t}$
are Janet reduced modulo
the passive difference system $\{ \, \tilde{f}_1 = 0, \, \ldots, \, \tilde{f}_s = 0 \, \}$.\label{de:differencesimple_c}
\end{enumerate}
\end{definition}

\begin{proposition}\label{prop:differencesimplemembership}
Let $\tilde{S}$ be a quasi-simple difference system over $\tilde{\cR}$
as in (\ref{eq:polypartialdifference}).
Let $E$ be the difference ideal of $\tilde{\cR}$
generated by $\tilde{f}_1$, \ldots, $\tilde{f}_s$
and let $Q$ be the smallest subset of $\tilde{\cR}$ which is multiplicatively closed,
closed under $\automorphism_1$, \ldots, $\automorphism_n$ and contains the initials $q_i := \init(\tilde{f}_i)$
for all $i = 1$, \ldots, $s$. Then a difference polynomial $\tilde{f} \in \tilde{\cR}$ is an
element of
\[
\begin{array}{rcl}
E : Q \! & \! = \! & \! \{ \, \tilde{f} \in \tilde{\cR} \mid (\theta_1(q_1))^{r_1} \ldots (\theta_s(q_s))^{r_s} \, \tilde{f} \in E\\[0.2em]
& &\quad \mbox{for some } \theta_1, \ldots, \theta_s \in \Mon(\setauto), \, r_1, \ldots, r_s \in \Z_{\ge 0} \, \}
\end{array}
\]
if and only if the Janet normal form
of $\tilde{f}$ modulo $\tilde{f}_1$, \ldots, $\tilde{f}_s$ is zero.
\end{proposition}

\begin{proof}
By definition of
$E : Q$, every element $\tilde{f} \in \tilde{\cR}$
for which Algorithm~\ref{alg:janetreducenonlin}
yields
Janet normal form zero is an element of $E : Q$.

Let $\tilde{f} \in E : Q$, $\tilde{f} \neq 0$.
Then there exist $q \in Q$
and $k_1$, ..., $k_s \in \Z_{\ge 0}$
and $c_{i,j} \in \tilde{\cR} \setminus \{ 0 \}$, $\alpha_{i,j} \in \Mon(\setauto)$, $j = 1$, ..., $k_i$,
$i = 1$, ..., $s$, such that
\begin{equation}\label{eq:linearcombalpha}
q \, \tilde{f} \, \, = \, \, \sum_{i=1}^s \sum_{j=1}^{k_i} c_{i,j} \, \alpha_{i,j}(\tilde{f}_i)\,.
\end{equation}
Among all pairs $(i, j)$ for which $\alpha_{i,j}$ involves a
non-admissible automorphism for $\tilde{f}_i = 0$ let the pair $(i^{\star}, j^{\star})$
be such that $\alpha_{i^{\star}, j^{\star}}(\ld(\tilde{f}_{i^{\star}}))$ is maximal
with respect to the ranking $\ranking$. Let $\automorphism$ be a non-admissible automorphism
for $\tilde{f}_{i^{\star}} = 0$ which divides the monomial $\alpha_{i^{\star},j^{\star}}$.
Since $\{ \tilde{f}_1 = 0, ..., \tilde{f}_s = 0 \}$ is passive, there exist $b \in Q$,
$l_1$, \ldots, $l_s \in \Z_{\ge 0}$
and $d_{i,j} \in \tilde{\cR} \setminus \{ 0 \}$
and $\beta_{i,j} \in \Mon(\setauto)$, $j = 1$, \ldots, $l_i$,
$i = 1$, \ldots, $s$, such that
\[
b \cdot (\automorphism \, \tilde{f}_{i^{\star}}) \, \, = \, \, \sum_{i=1}^s \sum_{j=1}^{l_i} d_{i,j} \, \beta_{i,j}(\tilde{f}_i)\,,
\]
where each $\beta_{i,j}$ involves only admissible automorphisms for $\tilde{f}_i = 0$.
Let $\gamma_{i^{\star}, j^{\star}} := \alpha_{i^{\star}, j^{\star}} / \automorphism$ and
multiply (\ref{eq:linearcombalpha}) by $\gamma_{i^{\star}, j^{\star}}(b)$ to obtain
\[
\gamma_{i^{\star}, j^{\star}}(b) \cdot q \, \tilde{f} \, \, = \, \,
\sum_{i=1}^s \sum_{j=1}^{k_i} c_{i,j} \cdot \gamma_{i^{\star}, j^{\star}}(b) \cdot \alpha_{i,j}(\tilde{f}_i)\,.
\]
In this equation we replace
\[
\gamma_{i^{\star}, j^{\star}}(b) \cdot \alpha_{i^{\star}, j^{\star}}(\tilde{f}_{i^{\star}}) \, \, = \, \,
\gamma_{i^{\star}, j^{\star}}(b \cdot \automorphism(\tilde{f}_{i^{\star}}))
\]
by
\[
\gamma_{i^{\star}, j^{\star}}\left(
\sum_{i=1}^s \sum_{j=1}^{l_i} d_{i,j} \, \beta_{i,j}(\tilde{f}_i) \right).
\]
Since $\gamma_{i^{\star}, j^{\star}} \, \beta_{i^{\star}, j^{\star}}$
involves fewer non-admissible automorphisms for $\tilde{f}_i = 0$ than $\alpha_{i^{\star}, j^{\star}}$,
iteration of this substitution process will rewrite equation (\ref{eq:linearcombalpha})
in such a way that every $\alpha_{i,j}(\ld(\tilde{f}_i))$ involving non-admissible automorphisms
for $\tilde{f}_i = 0$ will be less than $\alpha_{i^{\star}, j^{\star}}(\ld(\tilde{f}_{i^{\star}}))$
with respect to
$\ranking$. A further iteration of this substitution process
will therefore produce an equation as (\ref{eq:linearcombalpha}) with no $\alpha_{i,j}$
involving any non-admissible automorphisms for $\tilde{f}_i = 0$.

This shows that for every $\tilde{f} \in (E : Q) \setminus \{ 0 \}$ there
exists a Janet divisor of $\ld(\tilde{f})$ in the passive set defined
by $\tilde{f}_1 = 0$, \ldots, $\tilde{f}_s = 0$.
\end{proof}

Let $\Omega \subseteq \R^n$ be open and connected and fix $\mathbf{x} \in \Omega$.
Denoting the grid in (\ref{grid}) by $\Gamma_{\mathbf{x},h}$, we define
\[
\begin{array}{l}
\mathcal{F}_{\Omega,\mathbf{x},h} := \{ \, \tilde{u}\colon \Gamma_{\mathbf{x},h} \cap \Omega \to \C \mid
\mbox{$\tilde{u}$ is the restriction to $\Gamma_{\mathbf{x},h} \cap \Omega$ of}\\[0.2em]
\qquad \qquad \qquad \qquad \quad \mbox{some locally analytic function $u$ on $\Omega$} \, \}\,,
\end{array}
\]
and for a system $\tilde{S}$ of partial difference equations and inequations as in (\ref{eq:polypartialdifference})
we define the solution set
\[
\begin{array}{l}
\Sol_{\Omega,\mathbf{x},h}(\tilde{S}) := \{ \, \tilde{u} \in \mathcal{F}_{\Omega,\mathbf{x},h} \mid
\tilde{f}_i(\tilde{u}) = 0, \, \tilde{f}_{s+j}(\tilde{u}) \neq 0 \mbox{ for all}\\[0.2em]
\qquad \qquad \qquad \qquad \qquad \qquad i = 1, \ldots, s, \, j = 1, \ldots, t \, \}\,.
\end{array}
\]

\begin{definition}\label{de:differencedecomposition}
Let $\tilde{S}$ be a finite difference system over $\tilde{\cR}$ and $\ranking$ a ranking on $\tilde{\cR}$.
A \emph{difference decomposition} of $\tilde{S}$ is a finite collection of quasi-simple
difference systems $\tilde{S}_1$, \ldots, $\tilde{S}_r$ over $\tilde{\cR}$ such that
$\Sol_{\Omega,\mathbf{x},h}(\tilde{S}) = \Sol_{\Omega,\mathbf{x},h}(\tilde{S}_1) \uplus \ldots \uplus \Sol_{\Omega,\mathbf{x},h}(\tilde{S}_r)$.
\end{definition}

In the following algorithm, \algref{Decompose} in step~11
refers to an algorithm which computes a smallest superset of $G = \{ \tilde{f}_1, \ldots, \tilde{f}_s \}$ in $\tilde{\cR}$ that is Janet complete as defined on page~\pageref{alg:autoreducenonlin}
(see also Section~\ref{sec:janetdivision}).

\begin{algorithm}
\DontPrintSemicolon\SetCommentSty{textit}
\KwInput{A finite difference system $\tilde{S}$ over $\tilde{\cR}$, a ranking $\ranking$ on $\tilde{\cR}$, and a total ordering on $\setauto$ (used by \algref{Decompose})}
\KwOutput{A difference decomposition of $\tilde{S}$}
$Q \gets \{ \tilde{S} \}$; \, $T \gets \emptyset$\;
\Repeat{$Q = \emptyset$}{
  choose $L \in Q$ and remove $L$ from $Q$\;
  compute a decomposition $\{ A_1, ..., A_r \}$ of $L$, considered as an algebraic system, into quasi-simple systems (cf.\ Sect.~\ref{sec:simplealgebraicsystems})\;
  \For{$i = 1$, \ldots, $r$}{
     \If(\tcp*[f]{no equation and no inequation}){$A_i = \emptyset$}{
      \Return $\{ \emptyset \}$\;
    }
    \Else{
      $(a, G) \gets$ \algref{Auto-reduce}($A_i^{=}$, $\ranking$)\tcp*[r]{Alg.~\ref{alg:autoreducenonlin}}
      \If{$a = \text{\bf true}$}{
        $J \gets$ \algref{Decompose}($G$)\;
        $P \!\gets\! \{ \NF(\automorphism \tilde{f}, J, \ranking) \mid (\tilde{f}, \mu) \in J, \, \automorphism \in \overline{\mu} \}$\tcp*[r]{Alg.~\ref{alg:janetreducenonlin}}
        \If(\tcp*[f]{$J$ is passive}){$P \subseteq \{ 0 \}$}{
          replace each $\tilde{g} \! \neq \! 0$ in $A_i$ with $\NF(\tilde{g}, J, \ranking) \! \neq \! 0$\;
          \If{$0 \not\in A_i^{\neq}$}{
            insert $\{ \tilde{f} = 0 \mid (\tilde{f}, \mu) \in J \} \cup \{ \tilde{g} \neq 0 \mid \tilde{g} \in A_i^{\neq} \}$ into $T$\;
          }
        }
        \ElseIf{$P \cap \tilde{\cK} \subseteq \{ 0 \}$}{
          insert $\{ \tilde{f} = 0 \mid (\tilde{f}, \mu) \in J \} \cup \{ \tilde{f} = 0 \mid \tilde{f} \in P \setminus \{ 0 \} \} \cup \{ \tilde{g} \neq 0 \mid \tilde{g} \in A_i^{\neq} \}$ into $Q$\;
        }
      }
      \Else{
        insert\! $\{ \tilde{f} = 0 \mid \tilde{f} \in G \} \!\cup\! \{ \tilde{g} \neq 0 \mid \tilde{g} \in A_i^{\neq} \}$\! into $Q$\;
      }
    }
  }
}
\caption{\algref{DifferenceDecomposition}\label{alg:differencedecomp}}
\end{algorithm}

\begin{theorem}
Algorithm~\ref{alg:differencedecomp} terminates and is correct.
\end{theorem}

\begin{proof}
Algorithm~\ref{alg:differencedecomp} maintains a set $Q$ of difference systems that still have to be dealt with.
Given that termination of all subalgorithms has been proved, termination of Algorithm~\ref{alg:differencedecomp}
is equivalent to the condition that $Q = \emptyset$ holds after finitely many steps.

Apart from step~1, new systems are inserted into $Q$ in steps~18 and 20.
We consider the systems that are at some point an element of $Q$ as the vertices of a tree.
The root of this tree is the input system $\tilde{S}$. The systems which are inserted into $Q$ in
steps~18 and 20 are the vertices of the tree whose ancestor is the system $L$ that was
extracted from $Q$ in step~3 which in the following steps produced these new systems.
Since the for loop beginning in step~5 terminates, the degree of each vertex in the tree
is finite. We claim that every branch of the tree is finite, i.e., that the tree has finite
height, hence, that the tree has only finitely many vertices.

In case of step~20 the new system contains an equation which resulted from a non-trivial
difference reduction in step~9. When this new system will be extracted from $Q$ in a later
round, a decomposition into quasi-simple algebraic systems
will be computed in step~4. This may
produce new branches of the tree, but along any of these branches, after finitely many steps
the condition $a = $ true in step~10 will hold, because the order of the shifts in
leaders of the arising equations is bounded by the maximum order of shifts in leaders
of the ancestor system $L$.

In case of step~18 we are going to show that after finitely many steps a difference equation
is obtained whose leader has not shown up as a leader of an equation in any preceding
system in the current branch of the tree. First of all,
the passivity check (step~12) yielded an equation $\tilde{f} = 0$,
$\tilde{f} \in P \setminus \tilde{\cK}$, which is Janet reduced modulo $J$. Hence, either
$\ld(\tilde{f})$ is not contained in the multiple-closed set generated by $\ld(G)$, or
there exists $(\tilde{f}', \mu') \in J$ such that $\ld(\tilde{f}')$ is a Janet divisor of $\ld(\tilde{f})$,
but the degree of $\tilde{f}$ in $\ld(\tilde{f})$ is smaller than the degree of $\tilde{f}'$ in $\ld(\tilde{f}')$.
In the first case the above claim holds.
The second case cannot repeat indefinitely: First of all, if $\ld(\tilde{f}) = \ld(\tilde{f}')$, then
in a later round, either a pseu\-do-re\-duc\-tion of $\tilde{f}'$ modulo $\tilde{f}$ will be performed if
the initial of $\tilde{f}$ does not vanish, or $\init(\tilde{f}) = 0$ has been added as a new equation
(with lower ranked leader).
Since this leads to a sequence in $\Mon(\setauto)$ which strictly decreases, infinite
chains are excluded in this situation.
If case $\ld(\tilde{f}) \neq \ld(\tilde{f}')$ occurs repeatedly, then a sequence $((\theta_i \, \tilde{u}^{(\alpha)})^{e_i})_{i = 1, 2, 3, \ldots}$
of leaders of newly inserted equations arises,
where $\theta_i \in \Mon(\setauto)$, $\alpha \in \{ 1, \ldots, m \}$,
$e_i \in \Z_{\ge 0}$, such that $e_{i+1} < e_i$ holds (and where also $\theta_i \mid \theta_{i+1}$).
Any such sequence is finite. Hence, the first case arises after finitely many steps.
Therefore, termination follows from Dickson's Lemma.

In order to prove correctness, we note that a difference system is only inserted into $T$ if step~12 confirmed passivity. Such a system is quasi-simple as an algebraic system because (up to auto-reduction in step~9 and Janet completion in step~11) it was returned as one system $A_i$ in step~4.
Condition (\ref{de:differentialsimple_c}) in Definition~\ref{de:differencesimple} is ensured by step~14.
Hence, all difference systems in $T$ are quasi-simple.
Splittings of systems only arise in step~4 by adding an equation $\init(\tilde{f}) = 0$ and the corresponding inequation $\init(\tilde{f}) \neq 0$, respectively, to the two new systems replacing the given one.
Since no solutions are lost or gained, this leads to a partition as required by Definition~\ref{de:differencedecomposition}.
\end{proof}

%
%
\section{s-consistency check}\label{sec:Discretization}

Recall that $\tilde{S}^{=}$ (resp.\ $\tilde{S}^{\neq}$) denotes the set
of left hand sides of equations (resp.\ inequations) in
a difference system $\tilde{S}$. We shall use the same notation
for differential systems.

Clearly, if one approximates the partial derivatives occurring in a simple differential system $S$ by appropriate finite differences, then one obtains  a w-consistent approximation $\tilde{S}$ to $S$ (cf. Sect.~\ref{sec:consistency} and \ref{sec:NavierStokes}).

The following algorithm verifies s-consistency of such FDA.

\begin{algorithm}
\DontPrintSemicolon\SetCommentSty{textit}
\KwInput{A simple differential system $S$ over ${\cR}$, a differential ranking $\ranking$ on ${\cR}$, a difference ranking $>$ on $\tilde{\cR}$, a total ordering on $\setauto$ (used by \algref{Decompose}) and a difference system $\tilde{S}$ consisting of equations that are w-consistent with $S$}
\KwOutput{$\tilde{L} = \{ (\tilde{L}_1, b_1), ..., (\tilde{L}_r, b_r) \}$, where $\tilde{L}_i$ is s-consistent (resp.\ w-consistent) with $L_i\xleftarrow[h\rightarrow 0]{}  \tilde{L}_i$ if $b_i = \text{\bf true}$ (resp.\ {\bf false})}
$\tilde{L}=\{ \tilde{L}_1, \ldots, \tilde{L}_k \} \gets \algref{DifferenceDecomposition}(\tilde{S}^{=}, >)$\;
\For{$i = 1$, \ldots, $k$}{
    \If(\tcp*[f]{Def.~\ref{DefImplication}}){$\exists \tilde{f}\in \tilde{L}_i^{\neq}$ s.t. $\tilde{f}\rhd f\in \llbracket S^{=} \rrbracket$}{
      $\tilde{L} \gets \tilde{L}\setminus \{ \tilde{L}_i \}$\;
    }
    \Else{
      $b_i \gets \text{\bf true}$\;
      \For{$\tilde{f}\in \tilde{L}^{=}$}{
         compute $f \in {\cR}$ such that $\tilde{f} \rhd f$\tcp*[r]{Rem.~\ref{rem:continuouslimit}}
         \If(\tcp*[f]{Alg.~\ref{alg:janetreducenonlin}}){$\NF(f,S^{=},\ranking)\neq 0$}{
             $b_i \gets \text{\bf false}$; \, {\bf break}\;
         }
      }
    }
}
\Return $\{ \, (\tilde{L}_i, b_i) \mid \tilde{L}_i \in \tilde{L} \, \}$\;
\caption{\algref{S-ConsistencyCheck}\label{alg:discretization}}
\end{algorithm}

{\it Correctness} of the algorithm follows from Definition~\ref{def-wcons} (extended to inequations) and from  passivity of the output subsystems of Algorithm~\ref{alg:differencedecomp}. Their solution spaces partition the solution space of the input FDA. Thereby,
any subsystem $\tilde{L}_i$ in the output with $b_i=$ {\bf true} is s-consistent with
$L_i$, where 
$\tilde{L}_i\xrightarrow[h\rightarrow 0]{}  {L}_i$
and w-consistent if $b_i =$ {\bf false}.
If $b_i =$ {\bf true} for all $i$, then $\tilde{S}$ is s-consistent with $S$.
{\it Termination} follows from that of the subalgorithms.

%
%
\section{Illustrative Example}\label{sec:examples}

\begin{example}\label{IllustrativeExample}
We consider the system of nonlinear PDEs\\[-0.2cm]
\begin{equation}\label{eq:PDEsystem2}
\left\{ \begin{array}{rcl}
\frac{\partial u}{\partial x} - u^2 & = & 0\\[0.5em]
\frac{\partial u}{\partial y} + u^2 & = & 0\,,
\end{array} \right.
\qquad \qquad u = u(x,y)\,,
\end{equation}
which is a simple differential system,
as it is easily checked that the cross-derivative
$\partial_y (u_{x} - u^2) - \partial_x (u_{y} + u^2)$
reduces to zero modulo (\ref{eq:PDEsystem2}).
We investigate the discretized system which is obtained by replacing
$\partial_{x}$ and $\partial_{y}$ by
the forward differences $D_1^+$, $D^+_2$, respectively:
\begin{equation}\label{eq:FDA2a}
\left\{
\begin{array}{rcl}
D^+_1 \tilde{u} - \tilde{u}^2 & = & 0 \qquad \qquad (A)\\[0.5em]
D^+_2 \tilde{u} + \tilde{u}^2 & = & 0 \qquad \qquad (B)
\end{array}
\right.
\end{equation}
This system of nonlinear difference equations is simple as an algebraic system, but the
passivity check reveals the consequence
\[
\begin{array}{l}
\automorphism_2 A - \automorphism_1 B + (h \, \tilde{u}_{i+1,j} + h^2 \, \tilde{u}_{i,j}^2 + h \, \tilde{u}_{i,j} - 1) \, A \, +\\[0.5em]
\qquad (h \, \tilde{u}_{i,j+1} - h^2 \, \tilde{u}_{i,j}^2 + h \, \tilde{u}_{i,j} + 1) \, B \, \, = \, \, -2 \, h^3 \, \tilde{u}_{i,j}^4\,.
\end{array}
\]
The continuous limit of $\tilde{u}_{i,j}^4$ for $h \to 0$ is the differential polynomial $u^4$, which is not in the radical differential ideal corresponding to (\ref{eq:PDEsystem2}). Hence, FDA (\ref{eq:FDA2a}) is not s-consistent with system (\ref{eq:PDEsystem2}).

\noindent
Now we consider the discretization obtained by replacing
$\partial_{x}$ and $\partial_{y}$ by $D_1^+$ as before and the backward difference $D_2^-$, respectively:
\begin{equation}\label{eq:FDA2b}
\left\{
\begin{array}{rcl}
D^+_1 \tilde{u} - \tilde{u}^2 & = & 0 \qquad \qquad (C)\\[0.5em]
D^-_2 \tilde{u} + \tilde{u}^2 & = & 0 \qquad \qquad (E)
\end{array}
\right.
\end{equation}
In order to avoid negative shifts, we replace equation $(E)$ by $\automorphism_2 (E)$.
Then this system of nonlinear difference equations is simple because it is
algebraically simple and the passivity check yields
\[
\begin{array}{l}
\automorphism_1 E - (h^2 \, \tilde{u}_{i,j+1}^2 + h \, \tilde{u}_{i+1,j+1} + h \, \tilde{u}_{i,j+1} + 1) \, \automorphism_2 C \, +\\[0.5em]
\qquad C - E - h \, (h \, \tilde{u}_{i,j+1}^2 + \tilde{u}_{i,j+1} + \tilde{u}_{i,j}) \, E \, \, = \, \, 0\,.
\end{array}
\]
We conclude that FDA (\ref{eq:FDA2b}) is s-consistent with system (\ref{eq:PDEsystem2}).
\end{example}

%
%
\section{Navier-Stokes Equations}\label{sec:NavierStokes}
\begin{example}\label{Ex:NSE}
The Navier-Stokes equations for a three-di\-men\-sio\-nal incompressible viscous flow in vector notation are
\begin{equation}
 \dfrac{\partial \mathbf{u}}{\partial t}+(\mathbf{u}\cdot \nabla)\, \mathbf{u}
+\nabla  p -\frac{1}{\mathrm{Re}}\Delta\, \mathbf{u}=0\,,\quad \nabla\cdot \mathbf{u}=0\,, \label{NSE}
\end{equation}
where $\mathbf{x}=(x_1,x_2,x_3)$, $\mathbf{u}(\mathbf{x},t)$ is the velocity vector  $\mathbf{u}=(u,v,w)$, $p(\mathbf{x},t)$ is the pressure and   $\mathrm{Re}$ is the Reynolds number. For the ranking $\ranking_{{\rm TOP},{\rm lex}}$ (Example~\ref{ex:lexrankings}) such that
\begin{equation}\label{DifferentialRanking}
\partial_t \succ \partial_1 \succ \partial_2 \succ \partial_3\quad \mathrm{and}\quad p \succ u \succ v \succ w\,,
\end{equation}
the (non-admissible) prolongation $\nabla\cdot \partial_t\mathbf{u}=0$ of the right (continuity) equation in~\eqref{NSE} and its reduction modulo the left (momentum) equation yields the
pressure Poisson equation
\begin{equation}\label{PPE}
    \Delta\, p +\nabla\cdot (\mathbf{u}\cdot \nabla)\,\mathbf{u}=0\,,
\end{equation}
which is the {\em integrability condition}~(cf.~\cite{Seiler'10},\,p.50) to~\eqref{NSE}. Clearly, the differential system~\eqref{NSE} and~\eqref{PPE} satisfies the simplicity conditions (\ref{de:algebraicsimple_1})--(\ref{de:algebraicsimple_4}) in Definition~\ref{de:algebraicsimple}. Now we consider the following class of FDA to~\eqref{NSE} defined on the four-dimensional grid~\eqref{grid}
\begin{equation}\label{DNSE}
  D_t\mathbf{\tilde{u}}+(\mathbf{\tilde{u}}\cdot \mathbf{D})\, \mathbf{\tilde{u}}
+\mathbf{D}\,\tilde{p} -\frac{1}{\mathrm{Re}}\,{\tilde{\Delta}}\, \mathbf{\tilde{u}}=0\,,\quad \mathbf{D}\cdot \mathbf{\tilde{u}}=0\,,
\end{equation}
where $D_t$ approximates $\partial_t$, $\mathbf{D}=(D_1,D_2,D_3)$ approximates $\nabla$ and ${\tilde{\Delta}}$ approximates $\Delta$. It is clear that system~\eqref{DNSE} is w-consistent with
\eqref{NSE}. If one considers the difference analogue of ranking~\eqref{DifferentialRanking} satisfying
\begin{equation}\label{DifferenceRanking}
\sigma_t \succ \sigma_1 \succ \sigma_2 \succ \sigma_3\quad \mathrm{and}\quad \tilde{p} \succ \tilde{u} \succ \tilde{v} \succ \tilde{w}\,,
\end{equation}
then completion of
\eqref{DNSE} to a passive form by Algorithm~\ref{alg:differencedecomp} is equivalent to enlargement of this system with the integrability condition
\begin{equation}\label{DPPE}
(\mathbf{D}\cdot \mathbf{D})\,p+\mathbf{D}\cdot (\mathbf{\tilde{u}}\cdot \mathbf{D})\,\mathbf{\tilde{u}}=0\,.
\end{equation}
Eq.~\eqref{DPPE} approximates Eq.~\eqref{PPE} and can be obtained, in the full analogy with the differential case, by the prolongation $\mathbf{D}\cdot D_t\tilde{u}=0$ of the discrete continuity equation in system~\eqref{DNSE} and its reduction by the discrete momentum equation.

The left-hand sides of Eqs.~\eqref{NSE} and \eqref{PPE} form a difference \Gr basis of the ideal generated by
Eqs.~\eqref{DNSE} in $\Q(\mathrm{Re,h})\{\mathbf{\tilde{u}},\tilde{p}\}$. Hence, by Theorem \ref{th:s-cons}, FDA \eqref{DNSE}--\eqref{DPPE} to Eqs.~\eqref{NSE}, \eqref{PPE} is s-consistent.
\end{example}
\begin{remark}
Formulae~\eqref{DNSE} and~\eqref{DPPE} give s-consistent FDA of the Navier-Stokes and pressure Poisson equations in the two-di\-men\-sio\-nal case as well. Examples of such FDA were studied in~\cite{ABGLS'13}. One more s-consistent two-dimensional FDA was derived in~\cite{ABGLS'17}. In its approximation of Eq.~\eqref{PPE} the redundant to zero term $-\frac{\tilde{\Delta}\,(\nabla\cdot \mathbf{\tilde{u}})}{\mathrm{Re}}$ was included in the left-hand side of~\eqref{DPPE}. This inclusion improves the numerical behavior of FDA (cf.~\cite{Rempfer'06}, Sect.3.2).
\end{remark}

\begin{example}\label{Ex:NotSconsistentNSE}
For the two-dimensional system~\eqref{NSE}, \eqref{PPE} with grid velocities $(u,v)$ and pressure $p$ we consider the
discretization
\begin{equation}\label{FDA_to_NSE}
\left\lbrace
\begin{array}{l}
\tilde{e}^{(1)}:= D_1 \tilde{u} + D_2 \tilde{v}=0\,, \\[0.4em]
\tilde{e}^{(2)}:= D_t \tilde{u} + \tilde{u} D_1 \tilde{u}+\tilde{v}D_2\tilde{u}+D_1\tilde{p}-\frac{1}{\mathrm{Re}}\,{\tilde{\Delta}}\,\tilde{u}=0\,,\\[0.4em]
\tilde{e}^{(3)}:= D_t \tilde{v} + \tilde{u} D_1 \tilde{v}+\tilde{v}D_2\tilde{v}+D_2\tilde{p}-\frac{1}{\mathrm{Re}}\,{\tilde{\Delta}}\,\tilde{v}=0\,,\\[0.4em]
\tilde{e}^{(4)}:= {\tilde{\Delta}}\,\tilde{p} + (D_1\tilde{u})^2+2(D_2\tilde{u})\,(D_1\tilde{v})+(D_2\tilde{v})^2=0\,,
\end{array}
\right.
\end{equation}
where\\[-0.5cm]
\begin{equation*}
D_t=\frac{\sigma_t - 1}{h},\  D_{i}=\frac{\sigma_{i}-\sigma_{i}^{-1}}{2h},\  \tilde{\Delta}=\frac{\sigma_1+\sigma_2-4+\sigma_1^{-1}+\sigma_2^{-1}}{h^2}
\end{equation*}
and $i\in \{1,2\}$. Then FDA~\eqref{FDA_to_NSE} is w-consistent with~\eqref{NSE} and \eqref{PPE}. However, it is s-inconsistent since $\tilde{e}^{(4)}\not \in \llbracket \tilde{\cI}\rrbracket$ where $\llbracket \tilde{\cI}\rrbracket\subset \tilde{\cR}$ is the perfect closure (see Lemma 2.6) of the ideal  generated by $\{\tilde{e}^{(1)},\tilde{e}^{(2)},\tilde{e}^{(3)}\}$.
It follows,
as modulo $\llbracket \tilde{\cI}\rrbracket$ the equality holds
\[
\mathbf{D}\cdot (\mathbf{\tilde{u}}\cdot \mathbf{D})\,\mathbf{\tilde{u}}=(D_1\tilde{u})^2+2(D_2\tilde{u})\,(D_1\tilde{v})+(D_2\tilde{v})^2\,,
\]
whereas the difference operator $\mathbf{D}\cdot \mathbf{D}$
in~\eqref{DPPE} is not equal to $\tilde{\Delta}$:\\[-0.2cm]
\[
\mathbf{D}\cdot \mathbf{D}=\frac{\sigma_1^2+\sigma_2^2-4+\sigma_1^{-2}+\sigma_2^{-2}}{4h^2}\neq \tilde{\Delta}\,.
\]
\end{example}

%
%
\section{Conclusions}\label{sec:conclusions}

In this paper, for the first time, we devised a universal algorithmic approach to check s(trong)-consistency of a system of finite difference equations that approximates a polynomially nonlinear PDE system on a Cartesian solution grid. In our earlier paper~\cite{GR'10} we studied this problem for linear PDE systems and showed how to check their s-con\-sis\-ten\-cy by using differential and difference \Gr bases of ideals generated by the polynomials in PDE and FDA.
As this takes place, all related computations can be done, for example, with the Maple packages {\sc LDA}~\cite{GR'12} and {\sc Janet}~\cite{Maple-Janet'03}.

Extension of the \Gr basis method to the nonlinear case is not algorithmic due to the non-Noetherity of differential and difference polynomial rings. On the other hand, the differential Thomas decomposition (Def.~\ref{DifferentialThomasDecomposition}) and its difference analogue (Def.~\ref{de:differencedecomposition}) are fully algorithmic (cf.~\cite{GLHR'18,BGLHR'12,Robertz6} and Alg.~\ref{alg:differencedecomp}). These decompositions are essentials of the s-consistency check (Alg.~\ref{alg:discretization}). The differential Thomas decomposition is built into Maple 2018 and its implementation for previous versions of Maple is freely available on the web.  Algorithm~\ref{alg:differencedecomp} has not been implemented yet.

If we are looking for s-consistent FDA to a simple PDE system and for a (w-consistent) FDA Algorithm~\ref{alg:discretization} returns {\bf false}, as it takes place in Example~\ref{Ex:NotSconsistentNSE},  then we have to try another FDA and check the s-consistency again. In doing so, if we know a minimal generating set for the radical differential ideal generated by the input simple differential system, then its FDA should be tried as an input for Algorithm~\ref{alg:differencedecomp}. Such is indeed the case with the Navier-Stokes equations   (Ex.~\ref{Ex:NSE}), for which Algorithm~\ref{alg:differencedecomp} returns s-consistent discretization~\eqref{DNSE},\,\eqref{DPPE} if it is applied to Eqs.~\eqref{NSE} and ranking~\eqref{DifferenceRanking}.

However, the choice of FDA to the minimal generating set for the simple differential system as an input for Algorithm~\ref{alg:differencedecomp} not always yields s-consistent FDA,
as demonstrated by Example~\ref{IllustrativeExample}. In addition, designing an algorithm for construction of a minimal generating set for an ideal is an open problem for commutative polynomial rings and is probably unsolvable in the differential case.

In
applications of finite difference methods to
PDE systems which have integrability conditions, it is important not only to preserve these conditions at the discrete level, but to ensure also that
FDA is s-consistent with the PDE system.
FDA~\eqref{DNSE}, \eqref{DPPE} to the Navier-Stokes equations~\eqref{NSE} satisfies this requirement and for this reason it is  appropriate for numerical solution of initial or/and boundary-value problems for~\eqref{NSE} in the  velocity-pressure formulation.

\section{Acknowledgments}
The authors are grateful to the referees for their valuable remarks.
The contribution of the first author (V.P.G.) was partially supported by the Russian Foundation for Basic Research (grant No. 18-51-18005) and by the RUDN University Program (5-100).


\end{document}